\documentclass[11pt]{article}

\usepackage{amsfonts}
\usepackage{fancyhdr}
\usepackage{comment}
\usepackage{natbib}
\usepackage[a4paper, top=2.5cm, bottom=2.5cm, left=2.2cm, right=2.2cm]{geometry}
\usepackage{times}
\usepackage{amsmath}
\usepackage{multirow}
\usepackage{changepage}
\usepackage{colortbl}
\usepackage{amssymb}
\usepackage{graphicx}
\usepackage{tikzsymbols}
\newtheorem{theorem}{Theorem}

\newtheorem{corollary}[theorem]{Corollary}

\newenvironment{proof}[1][Proof]{\textbf{#1.} }{\ \rule{0.5em}{0.5em}}

\begin{document}

\title{Revisiting Fairness Impossibility with Endogenous Behavior}
\author{Elizabeth Maggie Penn and John W. Patty
\thanks{Departments of Political Science and Data \& Decision Sciences, Emory University, Atlanta GA, USA. Emails: \textit{elizabeth.m.penn@gmail.com}, \textit{jwpatty@gmail.com}}
}

\maketitle

\begin{abstract}
In many real-world settings, institutions can and do adjust the consequences attached to algorithmic classification decisions, such as the size of fines, sentence lengths, or benefit levels. We refer to these consequences as the \textit{stakes} associated with classification. These stakes can give rise to behavioral responses to classification, as people adjust their actions in anticipation of how they will be classified.  Much of the algorithmic fairness literature evaluates classification outcomes while holding behavior fixed, treating behavioral differences across groups as exogenous features of the environment. Under this assumption, the stakes of classification play no role in shaping outcomes.

\hspace{.15in}We revisit classic impossibility results in algorithmic fairness in a setting where people respond strategically to classification. We show that, in this environment, the well-known incompatibility between error-rate balance and predictive parity disappears, but only by potentially introducing a qualitatively different form of unequal treatment. Concretely, we construct a two-stage design in which a classifier first standardizes its statistical performance across groups, and then adjusts stakes so as to induce comparable patterns of behavior. This requires treating groups differently in the consequences attached to identical classification decisions. Our results demonstrate that fairness in strategic settings cannot be assessed solely by how algorithms map data into decisions. Rather, our analysis treats the human consequences of classification as primary design variables,     introduces normative criteria governing their use,  and shows that their interaction with statistical fairness criteria generates qualitatively new tradeoffs.    Our aim is to make these tradeoffs precise and explicit.

\end{abstract}

\bibliographystyle{apsr}

\section{Introduction}

Algorithmic systems are now used in a wide range of domains (\textit{e.g.}, criminal justice, hiring, auditing, benefits administration, \textit{etc}.).  Many of these systems are routinely evaluated with respect to statistical fairness criteria (\textit{e.g.}, error-rate balance and/or predictive parity across demographic groups).  Statistical criteria are attractive because they offer clear, ``context-free'' benchmarks.  For example, error-rate balance reflects the intuition that individuals who behave similarly should face similar chances of favorable and unfavorable decisions, while predictive parity captures the idea that a positive decision should carry the same informational meaning regardless of who receives it. While these systems vary widely in their details and use in the real world, the common statistical criteria used to evaluate them reflect a shared concern that algorithmic systems should not amplify existing inequities through their error patterns or through the meaning of their decisions.  In other words, algorithms and the data upon which they make their predictions should not unfairly disadvantage particular groups.

It is now well known that many of these fairness criteria are, unfortunately, fundamentally at odds with each other (\cite{KleinbergMullainathanRaghavan16}, \cite{Chouldechova17}).  A unifying theme of these impossibility results is that there are at least two margins that distinguish groups in ``algorithmic terms'': differences in the underlying ``prevalences'' (\textit{e.g.}, different \textit{base rates} of recidivism, default, or compliance) within the groups, and differences in the accuracy of the data collected about the groups' members (\textit{e.g.}, differences in the precision of standardized tests, differences in data collection, \textit{etc}.).  The ``impossibility'' reflected in many of these theoretical results often boils down to incompatibility of ``equalizing'' both of these margins for an arbitrary pair of groups.  Unsurprisingly, these  results have structured ensuing debates about fairness, shifting scholarly and policy attention from the question of ``which rule is best?'' to something closer to ``which tradeoff is acceptable?''

We believe that an important gap in the literature has stymied these debates.  Specifically, because the impossibility results are about \textit{statistical} notions of fairness, the context-free nature of these notions  provides little guidance for how to judge the acceptability of the trade-offs.  To make this concrete, there are many situations in which achieving error-rate balance will benefit one group, while achieving predictive parity will benefit a different group.  The statistical notions themselves are incapable of distinguishing between these two groups, because they do not specify whose errors matter, how harms scale with mistakes, or what behavioral responses would be induced by the algorithm itself.  One can see an analogy of this challenge by asking a statistician, ``which is worse, a Type-I error or Type-II error?'' 

 It is with this ambiguity in mind that we revisit these impossibility results in a setting where individuals respond strategically to algorithmic classification (\textit{i.e.}, in the presence of \textit{performativity}). In addition to making the discussion more ``social science friendly,'' incorporating strategic behavior by individuals forces us to be explicit about individuals' preferences (\textit{i.e.}, their incentives).  While this necessarily makes the resulting debate distinctly \textit{not} ``context-free,'' we believe coherence requires making the costs, benefits, and incentive effects of proposed fairness interventions explicit.   Once behavior is allowed to respond to the incentives induced by a classification algorithm, an ``escape route'' from the incompatibility identified by the impossibility results emerges.  This escape route---inducing all groups to have equal prevalences---is identified in the impossibility results themselves, and we are not the first to recognize this possibility.  However, sidestepping   statistical incompatibility does not make the underlying fairness tradeoffs disappear.  Instead, it reveals a new set of tradeoffs that only emerge once incentives are considered. More specifically, attaining joint satisfaction of statistical criteria such as error-rate balance and predictive parity requires that the algorithm offer different rewards or penalties to different groups when their members receive identical classification outcomes---a form of unequal treatment in consequences that has no analog in the standard impossibility framework. Our analysis therefore does not contradict the impossibility results, but shows how statistical and incentive-based fairness criteria interact  when behavior is endogenous to classification.

Our analysis departs from the canonical impossibility framework in two ways. First, we endogenize base rates by modeling behavior as a choice that responds to the  incentives created by classification.  This adjustment acknowledges that, in many real-world settings, individuals adjust their actions in anticipation of how they will be evaluated or classified by an algorithm (\textit{e.g.}, complying with regulations or improving one's observable characteristics). When behavior responds to classification in this way,  group prevalence becomes  an equilibrium object that is shaped by the algorithm itself. Within this model of performativity, we then explicitly treat the \textit{stakes}  of classification (\textit{e.g.}, financial rewards, fines, sentence lengths, audit intensity, \textit{etc.}) as a ``design variable'' within the definition of the algorithm itself. This move captures the reality that, in practice, many agencies, firms, and organizations can, and routinely do, adjust the severity of consequences attached to the classifications they apply to individuals.

Taken together, these departures expand the space of policy instruments beyond the classifier. Classical impossibility results ask what can be achieved by changing how signals are mapped into decisions while holding both behavior and stakes fixed. We instead ask what can be achieved when institutions are allowed to design the full incentive environment induced by classification. We construct a simple two-stage mechanism in which a classifier first standardizes its statistical performance across groups so as to equalize true and false positive rates via randomized post-processing of signals.  After achieving error-rate balance, the second stage of the mechanism adjusts the stakes of classification so as to induce similar behavior within the two groups. Under mild regularity conditions, this procedure yields identical confusion matrices across groups in equilibrium, and thus satisfies both error-rate balance and predictive parity. The cost of this  construction is that identical classification decisions may carry systematically different consequences across groups. In equilibrium, we can eliminate disparity in error rates and in decision meaning, \textit{but this requires possibly creating new disparities in the severity of consequences attached to identical decisions}.

This tradeoff highlights a tension that is often obscured in fairness debates. Error-rate balance and predictive parity are attractive precisely because they promise equal outcomes across groups: similar behavior leads to similar decisions, and similar decisions carry the same informational meaning regardless of group membership. A natural additional requirement is equal stakes: identical decisions should carry identical consequences regardless of group membership. Our results show that these goals are sometimes---but not always---in conflict, and that this conflict is distinct from the familiar tradeoffs among statistical fairness criteria. Institutions have long navigated this tension. Means-tested benefits and income-scaled penalties both reflect the principle that identical decisions may carry different consequences across groups. While the algorithmic fairness literature has largely overlooked the role of stakes as a design variable, it has recognized an analogous tension in classification rules, showing that remedying disparate impact can require differential treatment via group-specific rules. Our goal is not to resolve this tension or advocate for any particular remedy, but to identify precisely when and why equal treatment in consequences can conflict with equal statistical outcomes. We believe this is a necessary first step toward reasoned debate about when departures from these ideals are warranted.

More broadly, our contribution is to provide a tractable framework for studying feedback effects in algorithmic decision making. By modeling behavior, classification, and incentives within a unified equilibrium setting, we show how familiar fairness criteria interact once individuals' labels are no longer exogenous. This perspective complements existing work on strategic classification and performative prediction (\cite{HardtEtAl16}, \cite{PerdomoEtAl20}, \cite{PennPatty2025a}), while also moving the algorithmic fairness debate into a broader ``design space'' that explicitly accounts for the human consequences of algorithmic decisions.

\paragraph{Contributions.} We make three contributions in this article.  First, we develop a simple equilibrium framework in which group base rates are endogenous to algorithmic classification and arise from strategic responses to the incentives created by classification.  This framework departs from standard settings by treating individual behavior as a primary object of interest.  Second, within this framework,  we establish an existence result showing that the classical incompatibility between error-rate balance and predictive parity is an artifact of holding behavioral incentives fixed. Specifically, we provide a constructive two-stage mechanism demonstrating that these criteria can be jointly satisfied in equilibrium when stakes are treated as design variables, though only by potentially introducing a qualitatively different form of unequal treatment. Last, we introduce two incentive-based fairness criteria---\textit{equal stakes} and \textit{aligned incentives}---and fully characterize their interaction with error-rate balance and predictive parity. We show that while these four criteria cannot be jointly satisfied in general, there exists a broad set of circumstances in which all four \textit{are} simultaneously attainable. Together, our results show that accounting for incentive design and behavioral feedback can offer new perspectives on algorithmic fairness in strategic environments.

\section{Related Work \label{Sec:RelatedWork}} 

Our analysis relates to several connected literatures on algorithmic decision making. We build on classic    possibility and impossibility results in statistical fairness, work on post-processing and signal transformation, models of strategic classification and behavioral response, models of performative prediction and policy feedback, and recent discussions of how the severity of algorithmic decisions shapes fairness outcomes. This section situates our contribution within each of these literatures.

\subsection{Impossibility Results in Algorithmic Fairness}

A foundational line of work in algorithmic fairness establishes incompatibility results among widely used statistical fairness criteria.  \cite{KleinbergMullainathanRaghavan16} and   \cite{Chouldechova17} show that, except in degenerate cases, error-rate balance (or equalized odds) and predictive parity (or calibration-type conditions) cannot be simultaneously satisfied when groups differ in their underlying prevalences. These results are robust to allowing for randomized classifiers and post-processing, and they formalize a fundamental tension among fairness notions under fixed data-generating processes.  A key modeling assumption underlying these impossibility results is that group behavior---and hence underlying group prevalence---is exogenous and unaffected by the classifier. Behavioral differences across groups are treated as primitives of the environment rather than as objects that may respond to algorithmic decisions.

Our paper revisits these impossibility results by relaxing this central assumption. We model behavior as a direct choice over labels (\textit{e.g.}, compliance versus non-compliance), so that classification incentives shape the distribution of labels in a population. When behavior is endogenous to classification, group prevalence becomes an equilibrium object rather than a static input. The incompatibility identified by Kleinberg et al. and Chouldechova no longer binds because it is possible to manipulate behavior so as to equalize prevalence across groups.

\subsection{Possibility Results in Algorithmic Fairness}

A complementary literature asks whether the impossibility results can be circumvented by relaxing their underlying assumptions or expanding the design space. \cite{ReichVijaykumar21}    show that calibration and equal error rates can be reconciled by separately enforcing calibration on scores and equal error rates on the resulting classifiers, providing necessary and sufficient conditions for when such scores exist. \cite{hsu22}    develop a post-processing framework that approximately satisfies multiple fairness criteria simultaneously, translating tradeoffs among fairness definitions into a constrained optimization problem. 

\cite{JungEtAl20}    take a different approach, showing that when individuals' compliance choices respond endogenously to a classification rule, and when signal distributions are identical across groups, the classifier maximizing aggregate compliance (or minimizing overall crime, in their framework) naturally satisfies error-rate balance. This result emerges from the alignment between the designer's objective and the behavioral incentives the rule creates. Their result is not a direct reconciliation of error-rate balance and predictive parity, but demonstrates that endogenizing behavior can resolve apparent conflicts between classification objectives and fairness criteria. 

Our paper combines both mechanisms. We model endogenous behavior as in Jung et al., but allow for heterogeneous signal distributions across groups and treat stakes as an additional design variable. This yields a constructive existence result showing that joint satisfaction of error-rate balance and predictive parity is achievable, even when signal environments and cost distributions differ substantially across groups.

\subsection{Post-Processing and Equalized Error Rates}

\cite{HardtPriceSrebro16} show that error-rate balance  and its one-sided variant, equal opportunity, can be achieved through randomized post-processing of a fixed score. This result plays a central role in the fairness literature by demonstrating that certain fairness constraints can be imposed without retraining the underlying model or altering how information is extracted from data, beyond the information loss inherent in the constraint itself. Subsequent work has used post-processing as a standard tool for enforcing error rate constraints, often as a first step before considering additional objectives or trade offs.

Our analysis uses randomized post-processing in a similar spirit to equalize error rates across groups. However, this step is not our main contribution. Instead, it serves as a benchmark that isolates the remaining tension between error-rate balance and predictive parity under fixed behavior. The novelty of our approach lies in the subsequent stage,    {where we demonstrate that the remaining incompatibility can be overcome by adjusting the stakes of classification.}

\subsection{Strategic Classification and Unequal Costs of Response}

A growing literature studies classification when individuals respond strategically to decision rules. Early work by \cite{Dong18} models agents who manipulate features in response to a classifier.  \cite{HuEtAl19} extend this framework by allowing groups to face different costs of manipulation, highlighting how institutional inequalities can translate into disparate impacts even when a classifier is formally group-blind. \cite{MilliEtAl19} further study the welfare and social costs of strategic classification, emphasizing that robustness to manipulation can impose unequal burdens.    A complementary strand treats strategic responses as genuine behavioral change rather than feature manipulation. \cite{Shavit20}    study settings where agents' actions causally affect their true outcomes, showing that transparency can benefit decision-makers when gaming and genuine improvement are aligned.  \cite{PennPatty2025a}    extend the framework in \cite{JungEtAl20}    to characterize globally optimal classification rules under general designer objectives when behavior responds endogenously to classification, showing that while optimal rules can appear counterintuitive, they remain low-dimensional despite behavioral feedback.  Other recent work examines how fairness constraints interact with strategic behavior, showing that fairness interventions can alter incentives to manipulate features in unequal ways across groups (\cite{Zhang22}, \cite{Keswani23}). 

Like this literature, we model classification as a strategic environment in which individuals respond to the incentives created by algorithmic decisions. However, existing work on strategic classification primarily focuses on feature manipulation: individuals alter observable characteristics in order to cross a decision threshold, often with heterogeneous costs of manipulation across groups. In these models, strategic behavior affects the distribution of features conditional on label, but it does not change the underlying distribution of labels. As a result, group prevalences remain fixed, and the classical incompatibility between error-rate balance and predictive parity continues to apply.   Like \cite{JungEtAl20}  and  \cite{PennPatty2025a, PennPatty25AlgEndog}, our approach models strategic responses that change individuals' labels.  This distinction is central to our contribution: by allowing incentives to shape labels rather than just features, we show that behavioral responses can reconcile fairness criteria that are otherwise incompatible under fixed base rates. Related work on incentive-aware machine learning incorporates strategic behavioral responses to classification, and allows these responses to affect the probability of receiving a positive label  \cite{Podimata25}.  However, this literature does not study how incentive-induced changes in behavior alter group base rates, nor does it use incentive design to address the classical incompatibility between statistical fairness criteria. Our contribution is to place these behavioral responses at the center of the fairness analysis.

\subsection{Feedback, Performative Prediction, and Fairness}

Another related strand of work on performative prediction and outcome performativity (\cite{PerdomoEtAl20}, \cite{KimPerdomo23}, \cite{HardtMendler-Dunner23}) studies learning problems where model deployment changes the distribution of outcomes. These models formalize both feature and outcome performativity, whereby predictions reshape the data they aim to predict. Several papers connect performativity to fairness concerns and consider long-run or dynamic notions of fairness when interventions affect future disparities (\cite{Yin23}, \cite{Puranik22}).   Our model shares the core insight of this literature: algorithmic decisions can change the environment they are meant to predict. However, we focus on a specific and analytically tractable channel---behavioral responses to the stakes of classification---and on the implications for classical fairness trade-offs. 

In this respect our work aligns most closely with  \cite{Somerstep24}, who similarly observe that the classical incompatibility between error-rate balance and predictive parity disappears when group prevalences  are equal, and that performative feedback can be used to shape these prevalences.  Our analysis builds directly on this insight but takes a different tack. In the performative policy learning environment of Somerstep et al., the population---characterized by signal informativeness, costs of response, and prevalence rates---evolves endogenously under the policy, so that fairness can be achieved by steering the system toward a long-run state where response distributions become group-independent. Conversely, we take as given that some inequalities are persistent, structural, and not self-correcting through individual success or learning.   We treat the underlying signal informativeness and cost distributions of groups as static, and provide a constructive, two-stage mechanism that achieves fairness without relying on population drift to wash out cross-group differences. We show that this approach will yield identical confusion matrices across groups in equilibrium under mild regularity conditions, even when groups begin with very different signal distributions or cost environments. Our characterization makes the form of differential treatment that performative reform implicitly relies on explicit, and highlights  normative tradeoffs involved in using incentives as instruments of fairness.

\subsection{Consequences, Severity, and Fairness Beyond Classification}

Finally, a growing literature emphasizes that the fairness of algorithmic systems cannot be evaluated just in terms of binary decisions, and must account for the severity of the consequences those decisions impose. In many domains, algorithmic classifications determine not only whether an individual receives access or punishment, but also the magnitude of consequences such as fines, sentence lengths, loan terms, surveillance intensity, eligibility duration, or benefit levels (\cite{Eubanks18}, \cite{CorbettDaviesEtAl17}, \cite{Huq20}).  This literature reinforces the idea that the \textit{stakes} to  classification are central to the moral evaluation of algorithmic systems.   \cite{MunchBjerringMainz24} explicitly analyze the role of stakes in considering whether individuals have a right to explanation and procedural justification. They emphasize that stakes can arise not only from one-off high-impact decisions, but also from the cumulative effects of low-impact decisions.  \cite{PennPatty25AlgEndog}, in a related classification model, show that without constraints on stake-setting, classification outcomes can be channeled in any direction, highlighting the importance of normative constraints on how stakes can be set.

Related work on algorithmic reform and reparative justice emphasizes that addressing historical and structural inequalities often requires interventions that operate through differential remedies or compensatory mechanisms. This work argues that fairness may require targeted transfers, preferential access, or differential penalties designed to offset these inequities (\cite{davis21}, \cite{Binns18}). Our contribution isolates the mechanism through which such reform operates in classification settings: fairness is achieved by attaching systematically different consequences to identical classification decisions across groups. By isolating severity as a design variable, we hope to clarify some of the normative tradeoffs that are implicit in algorithmic reform in real-world settings.

\section{Model}
\label{sec:model}

We study a simple binary classification model in which individuals respond strategically to an algorithmic decision rule. The purpose of the model is not to capture any particular application in full detail, but to isolate the equilibrium interaction between classification rules, behavior, and fairness metrics in settings where structural disparities in signal noise or response costs exist.     {Throughout, we assume individuals choose their behavior in anticipation of classification. This is a natural and necessary timing assumption for any model of strategic behavior, since behavioral responses to classification are only possible if people can anticipate and respond to the classification rule. This fits settings where compliance is an ongoing or dispositional choice, such as tax compliance, regulatory adherence, and benefit eligibility, but not settings where the behavior of interest is fixed or has already occurred, such as past criminal history, medical diagnoses, or other immutable characteristics that cannot respond to classification incentives.} We hope that this simple framework can serve as a useful baseline for thinking about how to formalize strategic classification problems more generally.

\subsection{Environment\label{environment}}

There is a population of individuals indexed by $i$. Each individual belongs to a \textit{group} $g_i \in \{X,Y\}$. Group membership may affect the statistical properties of observed data and also the idiosyncratic behavioral costs faced by individuals.

Each individual chooses a binary \textit{behavior} $\beta_i \in \{0,1\}$. We interpret $\beta_i = 1$ as compliance with a norm, policy, or requirement (\emph{e.g.}, obeying the law, meeting a qualification threshold, truthfully reporting information), and $\beta_i = 0$ as non-compliance. We treat $\beta_i$ as the individual's true outcome (or ground truth label) that the algorithm seeks to predict. Choosing $\beta_i = 1$ incurs a private \textit{cost} $c_i \in \mathbb{R}$, which is observed by the individual but not by the algorithm. This cost may reflect financial, physical, psychological, legal, or social burdens associated with compliance.    {Lower values of $c_i$ correspond to lower costs of compliance, and individuals with $c_i<0$ naturally prefer compliant behavior in the absence of any extrinsic incentive from classification, while individuals with $c_i>0$ naturally prefer non-compliance. } Conditional on group membership,    {compliance} costs are independently drawn from a distribution with cumulative distribution function $H^g$ that is continuous on $\mathbb{R}$.

After the individual chooses behavior, a noisy \textit{signal}, or score, $s_i \in \mathbb{R}$ is generated. The signal is informative about behavior but imperfect. Specifically, for each group $g \in \{X,Y\}$, this signal is drawn from a behavior-dependent probability measure $f^g_{\beta}$ with cumulative distribution function $F^g_\beta$. We assume that for each group $g$, $f_1^g$ satisfies the strict monotone likelihood ratio property with respect to $f_0^g$. This simply means that higher signals always mean that it is more likely the individual chose behavior $\beta_i=1$, and is equivalent to the condition that $F_0^g(s)>F_1^g(s)$ for all signals $s$.

Last, an algorithm assigns each individual a binary \textit{decision} $d_i \in \{0,1\}$ based on the observed signal and group membership. We interpret $d_i = 1$ as a    {positive or favorable decision for person $i$ that carries direct benefits} (\emph{e.g.}, no audit, admission, approval) and $d_i = 0$ as an unfavorable decision (\emph{e.g.}, audit, rejection, denial).

\subsection{Algorithms}

Formally, an \textit{algorithm} is a (possibly randomized) mapping:
\[
\delta : \mathbb{R} \times \{X,Y\} \to [0,1],
\]
where $\delta(s,g)$ denotes the probability that an individual in group $g$ who generates signal $s$ receives decision $d=1$. We write $\delta^{g}$ for brevity, and we assume throughout that $\delta^g$ is measurable so that all expectations are well-defined. Our formulation allows the algorithm to potentially condition on group membership, but does not require it.    {We use the terms \textit{algorithm} and \textit{classifier} interchangeably.}

\subsection{Payoffs and Incentives}

Individuals care about both the classification decision they receive and the private cost incurred in choosing whether to comply. For an individual with cost $c_i$ in group $g_i$, \textit{utility} is given by:
\[
u(\beta_i,d_i,c_i,g_i) = r^{g_i} \cdot d_i - \beta_i \cdot c_i,
\]
where $r^{g}\in \mathbb{R}$ is the net benefit to an individual in group $g$ of receiving the favorable decision. We refer to $r^g$ as the \textit{stakes to classification} for group $g$. Importantly, we assume that $r^g$ may depend on group membership; this assumption allows us to consider environments in which the  consequences of a positive or negative classification  may differ across groups.\footnote{Our formulation implicitly assumes that an individual receives $r^g$ if assigned the positive classification outcome and 0 if assigned the negative outcome. Behaviorally, this is identical to the individual receiving $r^g_1$ from a positive classification outcome and $r^g_0$ from a negative classification outcome. With this latter formulation $r^g=r^g_1-r^g_0$. Consequently, $r^g$ is simply the net benefit the individual receives from positive classification, or the difference between a positive and negative classification outcome.}

Given an algorithm $\delta$ and signal distributions $F^g_\beta$, each individual chooses behavior to maximize expected utility. Because the signal is noisy and the algorithm maps signals into decisions, the expected probability of receiving a favorable decision depends on the chosen behavior, the signal distributions, the stakes to classification, and the algorithm.

\subsection{Behavioral Response and Endogenous Base Rates \label{remarkSection}}

Individuals anticipate how their behavior affects the distribution of signals and, through the algorithm, the probability of receiving a favorable decision. As a result, behavior responds strategically to the algorithm. An individual will choose to comply if compliance yields a higher expected payoff than non-compliance to the individual. Formally, compliance will be a function of the true positive and false negative rates of classification induced by an algorithm:
\begin{equation}
\begin{array}{lcr}TPR^g(\delta)&=&\int_{\mathbb{R}}\delta^g(s)f^g_1(s) ds,\\
FPR^g(\delta)&=&\int_{\mathbb{R}}\delta^g(s)f^g_0(s) ds.
\end{array}
\label{errors}
\end{equation}

\noindent With these terms in hand, individual $i$ belonging to group $g^i=g$ will choose to comply if and only if:
$$r^{g}\cdot TPR^{g}(\delta)-c_i\geq r^{g}\cdot FPR^{g}(\delta).$$

\noindent Consequently, for any algorithm $\delta$ and group $g$, optimal behavior takes a threshold form. There exists a group-specific cutoff:
\begin{equation}\label{eqThresh} \hat{c}^g(\delta)\equiv r^g\cdot\left(TPR^{g}(\delta)-FPR^{g}(\delta)\right),\end{equation}
such that individuals in group $g$ choose $\beta_i = 1$ if and only if: $$c_i \leq \hat{c}^g(\delta).$$ Given a distribution of costs $H^g(t)$, we define:
\begin{equation}\label{eqComply}
\pi^g(\delta) \equiv H^g(\hat{c}^g(\delta))
\end{equation}
as the \textit{equilibrium base rate} or \textit{prevalence} of compliance in group $g$ induced by algorithm $\delta$.

  In light of Equation \ref{eqComply}, two points are of note. First, if $TPR^{g}(\delta)=FPR^{g}(\delta)$, then equilibrium compliance will equal $H^g(0)$. In this case our classifier performs no better than random assignment, and individuals face no extrinsic incentive to comply. We term $H^g(0)$ \textit{sincere prevalence for group $g$}. If $TPR^{g}(\delta)>FPR^{g}(\delta)$, then $\delta$ performs strictly better than random assignment. We term such a classifier \textit{informative for group $g$}. If a classifier is informative for every group $g$, we refer to it simply as \textit{informative}.

Finally, we note that equilibrium base rates $\pi^g(\delta) $ are endogenous: they depend on the classifier itself, along with the signal distributions and the stakes of classification, $r^g$. This endogeneity is central to our analysis. Measures of algorithmic fairness are typically evaluated by conditioning on observed behaviors or outcomes. Here, both the distribution of behavior and outcomes are shaped by the algorithm. We now turn to fairness criteria in this setting.

\section{Fairness Criteria with Endogenous Behavior}

We start by formalizing several criteria of fair classification and interpreting them in a setting where behavior responds strategically to classification. We adopt standard statistical notions of fairness---error-rate balance and predictive parity---but note that these criteria are evaluated at equilibrium, with group prevalence determined endogenously by the classifier and the incentives it creates. We'll also distinguish between statistical fairness criteria that are evaluated on classification outcomes, and normative constraints on the design of incentives (we term these constraints \textit{equal stakes} and \textit{aligned incentives}). This distinction allows us to isolate the role of classification stakes  in    {reconciling statistical} fairness trade-offs.

\begin{itemize}
    \item \textbf{Error-rate balance} (\textit{i.e.}, equal false positive and false negative rates across groups) requires that the algorithm induce equal behavioral incentives across groups: individuals who differ only by group membership face identical tradeoffs when deciding whether to comply:
\[
\mathbb{E}[d|\beta=1, g]=\mathbb{E}[d|\beta=1, g^\prime]\,\,\,\text{ and }\,\,\,\mathbb{E}[d|\beta=0, g]=\mathbb{E}[d|\beta=0, g^\prime]\text{ for }g\not=g^\prime.
\]

In our framework, $$\mathbb{E}[d|\beta=b, g]=\int_{\mathbb{R}}\delta^{g}(s)f_\beta^g(s)ds,$$ and error-rate balance is equivalent to the requirement that: $$TPR^{g}(\delta^g)=TPR^{g^\prime}(\delta^{g^\prime})\,\,\,\text{ and }\,\,\,FPR^{g}(\delta^g)=FPR^{g^\prime}(\delta^{g^\prime})\text{ for }g\not=g^\prime.$$ In the previous section we showed that the true positive and false positive rates induced by a classifier are the primary objects governing equilibrium behavior. As shown in Equation \ref{eqThresh}, individuals' incentives to comply depend only on the difference between these rates and the stakes to classification. Error rates play a central role in \textit{determining} equilibrium prevalence but are defined independently of group prevalence---they depend only on the classifier and the signal structure.

    \item \textbf{Predictive parity} (\textit{i.e.}, equal calibration across groups) requires that the algorithm equalize the information content of decisions after classification:
    $$\mathbb{E}[\beta|d=1,g]=\mathbb{E}[\beta|d=1, g^\prime]\text{ for }g\not=g^\prime.$$
    
In our framework, 
\[
\mathbb{E}[\beta|d=1, g]=\frac{\pi^g(\delta^g)\cdot TPR^{g}(\delta^{g})}{\pi^g(\delta^g)\cdot TPR^{g}(\delta^{g})+(1-\pi^g(\delta^g))\cdot FPR^{g}(\delta^{g})},
\]
where $\pi^g(\delta^g)$ is the equilibrium base rate of compliance induced by algorithm $\delta^g$: $$\pi^g(\delta^g)=H^g\big(r^g\cdot (TPR^{g}(\delta^{g})-FPR^{g}(\delta^{g}))\big).$$

Unlike error-rate balance, predictive parity \textit{does} depend on group behavior.  In our setting, prevalence is determined in equilibrium by the incentives an algorithm induces. While predictive parity can, in principle, be achieved through classifier design, doing so is generally incompatible with simultaneously equalizing error rates when groups differ in  prevalence. Addressing this tension will require us to control the behavioral responses that determine equilibrium prevalence.

    \item \textbf{Equal stakes} (\textit{i.e.}, equal consequences of classification across groups) requires that the net rewards and penalties to classification be the same across groups:
    $$r^g=r^{g^\prime} \text{ for }g\not=g^\prime.$$
    
    Equal stakes requires that identical classification decisions carry the same consequences across groups. This condition rules out using differential incentives as a fairness instrument, and serves as a benchmark that helps us isolate what can be achieved through classifier design alone.
    
    \item \textbf{Aligned incentives} (\textit{i.e.}, positive behavioral consequences of classification) requires that the classifier not strictly disincentivize compliant behavior for any group:
    $$\pi^g(\delta^g)\geq H^g(0)\text{ for all groups }g.$$
    
Aligned incentives requires that classification weakly increases compliance within each group relative to baseline behavior. This condition rules out mechanisms that satisfy statistical criteria by penalizing compliance. Consequently, it ensures that classification does not create perverse behavioral incentives in equilibrium.\footnote{
For informative classifiers (those for which $TPR^g(\delta)>FPR^g(\delta)$) aligned incentives rules out setting stakes $r^g<0$. These negative stakes would strictly disincentivize non-compliance by effectively rewarding non-compliance or penalizing compliance.}

\end{itemize}

    Error-rate balance and predictive parity capture two distinct fairness criteria that are widely invoked in algorithmic decision making. Error-rate balance reflects the idea that people who behave the same way should face the same chance of favorable and unfavorable decisions, regardless of their group membership. Predictive parity reflects the complementary concern that the informational meaning of a favorable decision should be the same across groups, so that receiving a positive classification outcome conveys the same information regardless of who receives it. Together, they ensure that a classifier doesn't penalize certain groups through asymmetric error patterns or devalue decisions by attaching different informational content to the same outcomes. 
    
Equal stakes and aligned incentives reflect normative constraints on how incentives can be used when behavior responds to classification. Equal stakes captures the concern that the same algorithmic decisions should not carry systematically different human consequences across groups. Aligned incentives imposes a directional constraint, requiring that classification encourage rather than discourage compliant behavior. Together, they  restrict the use of incentives to achieve statistical fairness by ruling out differential rewards or punishment, or by encouraging socially harmful behavior.

\section{Satisfying Statistical Fairness Criteria by Adjusting Stakes \label{resultSection}}

The well-known impossibility theorems of \cite{KleinbergMullainathanRaghavan16} and \cite{Chouldechova17} show that when groups have different  base rates of compliance it is impossible to design a classification rule that simultaneously attains predictive parity and error-rate balance. Our framework suggests a potential workaround by differentially manipulating behavioral responses to classification in order to equalize prevalence across groups. Theorem \ref{impossible} shows that, under mild conditions, whenever the severity of consequences can be adjusted across groups, it is always possible---regardless of differences in signal noise or compliance costs---to construct a classification rule that attains both error-rate balance and predictive parity in equilibrium.    {While all formal proofs are contained in Appendix} \ref{Sec:Proofs},    {our proof strategy is to construct this rule in two-stages, exploiting the mechanics of how classification affects behavior. First, because error rates depend only on the classifier and signal structure---and not on group prevalence---they can be equalized through signal post-processing without referencing equilibrium behavior or incentives. We can then adjust the stakes of classification to shape the prevalence of compliance that those error rates induce in equilibrium. This second stage construction isolates the tradeoffs that must be made in order to achieve both error-rate balance and predictive parity via incentives.}

\begin{theorem}
\label{impossible} 
For any two groups $X$ and $Y$, there exists an informative classification rule $\delta$ and a system of classification stakes $(r^X, r^Y)$ satisfying aligned incentives such that predictive parity and error-rate balance are jointly satisfied in equilibrium.
\end{theorem}
\noindent Theorem \ref{impossible} shows that, with endogenous behavior, unequal classification stakes are always sufficient to reconcile error-rate balance and predictive parity in equilibrium without sacrificing aligned incentives. The following corollary  characterizes the situations in which unequal stakes are also \textit{necessary} to reconcile these fairness goals.

\begin{corollary} Fix any informative classification rule that induces error-rate balance across groups. If the distribution of compliance costs for group $Y$ stochastically dominates that of group $X$ then predictive parity can be attained only by allowing the stakes of classification to differ across groups. Otherwise, there exists a classification rule and a system of stakes satisfying equal stakes, error-rate balance, and predictive parity, though potentially at the cost of violating aligned incentives.
\label{corollary}
\end{corollary}

   {Corollary} \ref{corollary}    {establishes that equal stakes, error-rate balance, and predictive parity can be jointly satisfied whenever the two groups' cost distributions are not stochastically ordered. Whether aligned incentives can also be preserved depends on an additional condition: all four criteria can be simultaneously satisfied if and only if there exists} $\bar{c}\geq 0$    {such that} $H^X(\bar{c})=H^Y(\bar{c})$.    {In words, the two groups' cost CDFs must cross at a positive value. The condition} $\bar{c}\geq 0$    {is what ensures that equal stakes can induce equal prevalence at a positive compliance threshold, preserving aligned incentives. This condition is not knife-edged: if two cost distributions satisfy this condition, then any sufficiently small perturbation of either distribution will preserve a crossing near} $\bar{c}$.    {The set of cost environments in which all four criteria can be simultaneously satisfied therefore has positive measure.}

\subsection{Discussion\label{discussSection}}

Theorem \ref{impossible} and Corollary \ref{corollary}  highlight the tradeoffs that arise once algorithmic stakes are treated as design variables. Theorem \ref{impossible} shows that error-rate balance, predictive parity, and aligned incentives can be jointly achieved, but only by potentially relaxing equal stakes, allowing the consequences of identical decisions to differ across groups. Corollary \ref{corollary} characterizes the situations in which a violation of equal stakes is necessary in order to satisfy other fairness goals. When one group's cost distribution first-order stochastically dominates another's---a standard way of formalizing structural disadvantage across groups---predictive parity and error-rate balance can never be jointly realized under equal stakes. When cost distributions are \textit{not} ordered by stochastic dominance, equal stakes can always be preserved, but potentially at the expense of aligned incentives, requiring incentive schemes that penalize compliant behavior.     {As an existence proof, the mechanism we construct to achieve joint satisfaction of fairness goals prioritizes analytical tractability over optimality. Our construction demonstrates what is achievable, but it is not intended as a prescriptive recommendation for any particular application.}

A related distinction is also worth drawing. While individual behavior depends on both the consequences of a positive or negative classification ($r$) along with the costs of behavioral change ($c_i$), the policy instruments that generate these costs and consequences are not interchangeable. In our framework, we implement differential incentives through the stakes attached to algorithmic decisions, rather than through direct interventions that reshape private costs of compliance, which may be unobservable or infeasible to alter. This distinction matters: adjusting the severity of classification outcomes is a natural and widely used policy lever, whereas equalizing underlying cost distributions across groups would require more challenging forms of intervention. 

Finally, throughout we have assumed that people value the favorable classification outcome independently of its accuracy. In other words, people prefer $d=1$ regardless of whether that decision correctly reflects their behavior. These preferences generate the compliance incentives at the heart of our proof of Theorem \ref{impossible}.    We could alternatively think of a more general class of preferences in which people prefer different kinds of classification outcomes (such as valuing accurate classification). While we leave richer preference specifications to future work, we note that for many different preference environments the resulting compliance incentives can similarly be expressed in terms of the true and false positive rates induced by an algorithm, and our proof strategy would carry through. A full characterization of when all four fairness criteria can be jointly satisfied under more general preferences is left to future work.

\section{Illustrating Fairness Tradeoffs via Stakes}

   {We now present three examples to illustrate how adjusting the stakes of classification can enable joint satisfaction of statistical fairness criteria that are otherwise incompatible.} In the first, group $X$'s costs of compliance are stochastically dominated by those in group $Y$, and    {achieving both} error-rate balance and predictive parity requires violating equal stakes. In the second and third examples, neither group stochastically dominates the other, and we are always able to simultaneously satisfy error-rate balance, predictive parity, and equal stakes. However, whether we can also attain aligned incentives depends on the specifics of the group cost distributions and whether equal-stakes fairness is achieved by ``lifting people up'' through increased incentives or ``leveling people down'' via the suppression of compliant behavior.  In the examples that follow we term a group $A$ \textit{advantaged} and $B$ \textit{disadvantaged} if sincere compliance---the fraction of the population with negative costs of compliance---is strictly higher in $A$ than in $B$.

\paragraph{Example 1: Fairness via differential treatment.}
Let compliance costs for group $Y$ be distributed $N[1, 1]$ and those for group $X$ be distributed $N[0, 1]$, so that members of $Y$ face uniformly higher costs of compliance ($Y$ is disadvantaged relative to $X$). After post-processing signals to equalize error rates, let $\mathcal{T}$ be the true positive rate for both groups and $\mathcal{F}$ be the false positive rate.\footnote{$\mathcal{T}$ and $\mathcal{F}$ are constructed in Appendix \ref{Sec:Proofs}. For the purposes of our examples we simply require that they exist and that $\mathcal{T}>\mathcal{F}$, so that our classifier is informative (which our construction guarantees).} Finally, let $\mathcal{E}=\mathcal{T}-\mathcal{F}>0$ be the difference in these true positive and false positive rates. 

By Equation \ref{eqThresh}, we know that individuals in group $g$ with costs below cutoff $\hat{c}^g$ will comply, with $\hat{c}^g=r^g\cdot \mathcal{E}.$
In this case, equal stakes necessarily induces lower equilibrium compliance in group $Y$, and hence lower positive predictive value. Predictive parity therefore cannot be achieved under equal stakes. Restoring predictive parity requires offering stronger incentives to disadvantaged group $Y$. Setting $r^X=0$ and $r^Y=\frac{1}{\mathcal{E}}$ yields identical prevalence across the groups of $50\%$. Aligned incentives is also satisfied, as prevalence in $X$ is unchanged from sincere prevalence and prevalence in $Y$ increases from 16\% to 50\%. 

\begin{figure}[h!]
\centering
\includegraphics[width=3in]{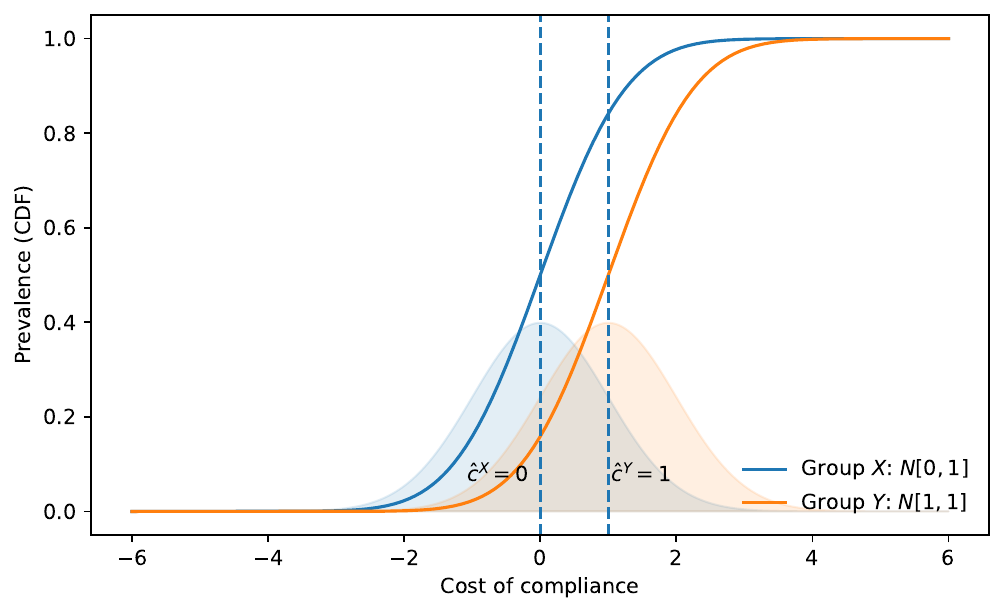}
\caption{Equalizing prevalence requires differential stakes.}
\label{ex1Fig}
\end{figure}

Figure \ref{ex1Fig} depicts the cumulative distribution functions and underlying probability densities of compliance costs for the two groups, along with the equilibrium cutoffs in costs, $\hat{c}^g$ below which members of $g$ will choose to comply.  In this case, one group's cost distribution stochastically dominates the other's, so equal stakes necessarily induce lower equilibrium compliance for the disadvantaged group. This (global) structural disadvantage makes some form of differential treatment unavoidable if statistical fairness is to be achieved.

\paragraph{Example 2: Fairness via lifting up.}
Now let compliance costs for group $X$ be distributed $N[0, 2]$ and those for group $Y$ be distributed $N[1, 1]$, so that $X$ is advantaged relative to $Y$. In this case, the CDFs of the groups' cost distributions uniquely cross at $\overline{c}=2$. By setting stakes equal to $r^X=r^Y=\frac{2}{\mathcal{E}}$ we can equalize prevalence across groups to $H^X(2)=H^Y(2)=84\%$. Here, we are able to simultaneously satisfy all four fairness criteria, demonstrating that the ability to satisfy all four criteria is not knife-edged. Figure \ref{ex2Fig} again depicts the CDFs and underlying PDFs of compliance costs for the two groups.

\begin{figure}[h!]
\centering
\includegraphics[width=3in]{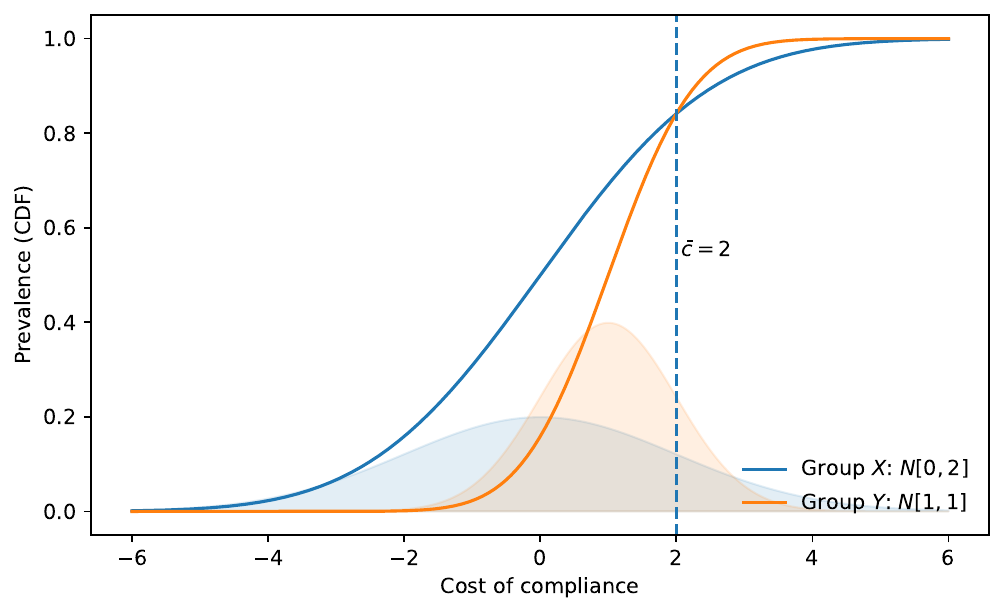}
\caption{Equalizing prevalence requires lifting up disadvantaged group $Y$.}
\label{ex2Fig}
\end{figure}

In this example, inequality is concentrated among a subset of higher-cost individuals in disadvantaged Group $Y$ centered at $c=1$. With sufficiently strong incentives, group $Y$ is lifted into equal compliance with $X$, with the resulting behavioral changes driven primarily by increased compliance within $Y$. However, our next example demonstrates that equal stakes may also require a form of leveling down.

\paragraph{Example 3: Fairness via leveling down.}
Now let compliance costs for group $X$ be distributed $N[0, 2]$ and those for group $Y$ be distributed $N[-1, 1]$.  In this case $Y$ is the advantaged group, and the CDFs of the groups' cost distributions uniquely cross at $\overline{c}=-2$. Here, equal stakes and predictive parity can be jointly achieved only by  ``flipping the stakes,'' or setting a negative $r^X=r^Y=-\frac{2}{\mathcal{E}}$ that \textit{penalizes} compliance rather than rewarding it.  A consequence is that equilibrium compliance falls below its sincere level in both groups, to $16\%$---a violation of aligned incentives. Again, Figure \ref{ex3Fig}  depicts the CDFs and underlying PDFs of the cost distributions for the two groups.

\begin{figure}[h!]
\centering
\includegraphics[width=3in]{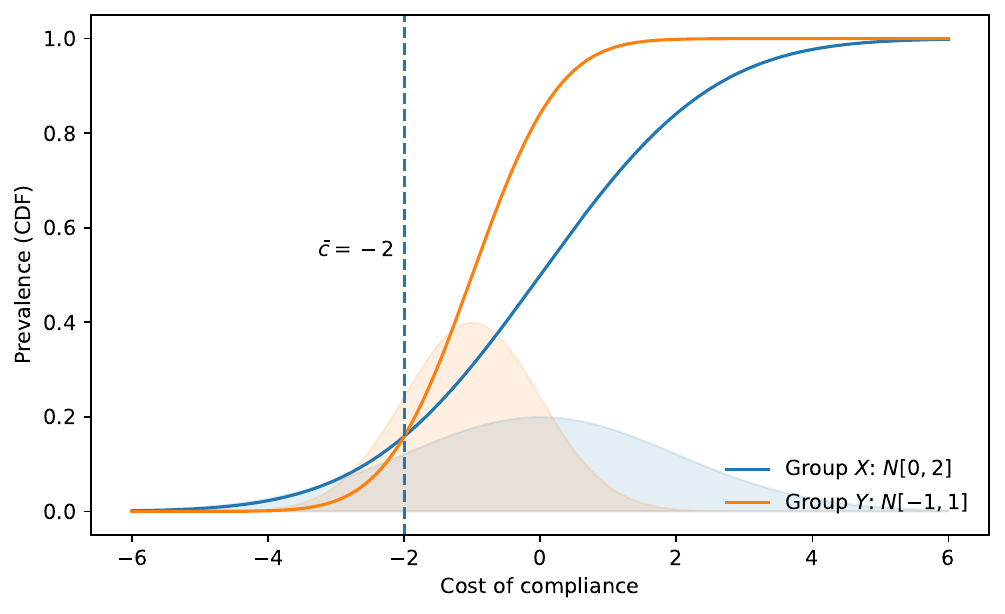}
\caption{Equalizing prevalence requires leveling down advantaged group $Y$.}
\label{ex3Fig}
\end{figure}

In this example inequality is concentrated among a low-cost, advantaged subset of individuals in Group $Y$ centered at $c=-1$.  Preserving equal stakes comes at the cost of suppressing compliance for everyone, with the largest behavioral change now borne by the advantaged Group $Y$.

\section{Conclusion}
\label{sec:conclusion}

We've developed a tractable framework for evaluating fairness interventions in settings where individuals respond strategically to algorithmic classification and where these responses depend on the stakes attached to classification outcomes. Such settings include environments structured by fines, benefits, sentences, eligibility rules, and institutional penalties borne by people. We show that in these environments,    {incompatibilities between statistical notions of fairness can be overcome} by expanding the set of design variables under consideration. When behavior responds endogenously to classification, and the stakes of classification are allowed to vary, classifiers can always be designed to satisfy error-rate balance and predictive parity, along with a behavioral fairness criterion we term ``aligned incentives.'' However, this comes at the cost of shifting disparity away from the \textit{distribution} of classification outcomes and toward the \textit{consequences} attached to those outcomes.

When fairness is pursued in strategic environments, designers must confront how behavioral incentives are structured and who bears the burden of behavioral change. This requires explicit choices about whether identical decisions should carry identical consequences across groups (equal stakes), whether classification should weakly encourage socially desirable behavior (aligned incentives), and whether to enforce statistical criteria such as error-rate balance and predictive parity.  Our results show that these objectives are not jointly attainable in general. However, we also show that in some environments---specifically, when group cost distributions cross at positive values---it \textit{is} possible to satisfy all four criteria simultaneously. Importantly, this compatibility is not knife-edged: it holds over an open set of cost distributions. Consequently, while tradeoffs between these fairness criteria are real, they are not inevitable.

Our framework deliberately models stakes in their simplest form, as the net payoff difference between positive and negative classification outcomes. This minimalism isolates the behavioral channel through which fairness criteria interact, but raises a new set of important normative questions. How rewards and penalties are structured, whether gains and losses are treated symmetrically, and how the distributive consequences of classification fall across people and groups are all questions our framework is designed to accommodate but does not resolve. A goal of ours has been to provide a flexible framework in which future work considering punishment severity, asymmetry,  and the distributive consequences of algorithmic systems can be undertaken via the equilibrium analysis of behavior.

Finally, our results connect directly to the disparate treatment/disparate impact distinction familiar from antidiscrimination law and the algorithmic fairness literature. A recurring theme in this literature is that remedying disparate impact often requires disparate treatment via rules that condition on group membership (\cite{CorbettDaviesEtAl17}).    Our results show that this same tension extends to the consequences attached to decisions, not just the decision rules themselves. Specifically, we characterize exactly when requiring equal stakes (an equal treatment condition) necessarily produces disparate impact in equilibrium outcomes. \textit{Griggs v. Duke Power Co.} (1971) holds that a formally neutral rule requires justification when it produces disparate impact, and acknowledges that disparate treatment may be warranted as a remedy in such cases. Our framework does not resolve whether differential stakes are ever the right remedy---such a conclusion would depend on institutional goals, legal constraints, and moral commitments that lie outside the scope of any one model. But by identifying precisely when and why equal stakes produce disparate impact, we hope to contribute to the informed debate that such judgments require.

\bibliography{john-impossibility}

\newpage
\appendix
\section{Proofs \label{Sec:Proofs}}

\noindent \textbf{Theorem \ref{impossible}}
\textit{For any two groups $X$ and $Y$, there exists an informative classification rule $\delta$ and a system of classification stakes $(r^X, r^Y)$ satisfying aligned incentives such that predictive parity and error-rate balance are jointly satisfied in equilibrium.}

\ \\ \begin{proof}
We consider two groups $g\in\{X, Y\}$ with signal distributions $f_1^g$, $f_0^g$ and cost distributions $H^g$. The proof proceeds in two steps. In Step 1 we will construct a group-specific classifier, $\delta^g$, that will guarantee that  true and false positive rates---$TPR^g(\delta^g)$ and $FPR^g(\delta^g)$---are equalized across groups, and consequently that error-rate balance is satisfied. This classifier introduces controlled randomization in its final decision to equalize these error rates. 

Because group-level incentives to choose $\beta_i=1$ depend on the product $r^g\cdot (TPR^g(\delta^g)-FPR^g(\delta^g))$, manipulating base rates via stakes $r^g$ requires the classifier to generate error rates such that $TPR^g(\delta^g)\not=FPR^g(\delta^g)$. Our construction in Step 1 ensures that $TPR^g(\delta^g) >FPR^g(\delta^g)$, so that our classifier is informative.  It's known that equal error rates can be attained via randomized post-processing of scores (Hardt et al., 2016). We provide a simple construction that demonstrates exactly how this post-processing interacts with endogenous behavior and incentive design.

In Step 2 we  adjust the stakes to classification, $r^g$, for each group so that induced prevalence is the same in both groups. This adjustment equalizes base rates of compliance across groups while preserving the common error rates established in Step 1. It then follows that error-rate balance and predictive parity are simultaneously satisfied by classifiers $\delta^g$ and reward structure $r^g$. Moreover, the induced confusion matrices will be identical across groups.

\ \\ \textbf{Step 1.}
For each group $g$, pick a score threshold $s^g$ with $0<F^g_1(s^g)<F^g_0(s^g)<1$.\footnote{Such a signal $s^g$ must exist because $F^X$ and $F^Y$ possess full support on $\mathbb{R}$.}  Define two numbers: $$\ell^g=F^g_0(s^g)\,\,\,\text{ and }\,\,\,m^g=F^g_1(s^g).$$ In words, $\ell^g$ is the probability a type $\beta=0$ from group $g$ scores below $s^g$, and $m^g$ is the probability a type $\beta=1$ from group $g$ scores below $s^g$. We have $m^g<\ell^g$ by the MLRP. Now we create two group-specific classifiers over the signals that will depend on (to be defined) constants $a^g, b^g\in[0, 1]$:

\begin{equation}\delta^g(s)=\left\{\begin{array}{ll}
a^g&\text{ if }s<s^g,\\
b^g&\text{ if }s\geq s^g.\\
\end{array}\right. \label{possibleClass}\end{equation}
We must now choose $a^g, b^g$ to equalize error rates across groups. Note that for each group our true positive rates  (TPR) and false positive rates (FPR) are:

$$\begin{array}{rclc}
\mathbb{E}[\delta^g(s)|\beta=1]&=&a^g\cdot m^g+b^g\cdot(1-m^g),&\hspace{.5in}(TPR^g(\delta^g))\\
\mathbb{E}[\delta^g(s)|\beta=0]&=&a^g\cdot \ell^g+b^g\cdot(1-\ell^g).&\hspace{.5in}(FPR^g(\delta^g))\\
\end{array}$$
To  equalize these rates across groups we fix $TPR^g(\delta^g)=\mathcal{T}$ and  $FPR^g(\delta^g)=\mathcal{F}$ and solve:
$$\begin{array}{rcl}
\mathcal{T}&=&a^g\cdot m^g+b^g\cdot(1-m^g),\\
\mathcal{F}&=&a^g\cdot \ell^g+b^g\cdot(1-\ell^g),\\
\end{array}$$
yielding the unique solution: \begin{equation}\label{possibleParameters}\begin{array}{rcl}
a^g&=&\dfrac{\mathcal{F}\cdot(1-m^g)-\mathcal{T}\cdot(1-\ell^g)}{\ell^g-m^g},\\
&&\\
b^g&=&\dfrac{\mathcal{T}\cdot \ell^g-\mathcal{F}\cdot m^g}{\ell^g-m^g}.
\end{array}\end{equation} It remains to show that there are feasible choices of $\mathcal{T}, \mathcal{F}$ that ensure $a^g, b^g\in[0, 1]$, so that our classifier maps signals into probabilities over decisions.  We accomplish this by constructing a particular choice of $\mathcal{T}$ and $\mathcal{F}$ that is both informative ($\mathcal{T}>\mathcal{F}$) and feasible ($a^g, b^g\in[0, 1]$), setting:
\begin{equation}\begin{array}{rcl}
\mathcal{F}&=&\dfrac{1}{2},\\
&&\\
\mathcal{T}&=&\dfrac{1}{2}+\min\limits_{G}\left(\dfrac{\ell^g-m^g}{2\cdot \max\{\ell^g, 1-\ell^g\}}\right).\\
\end{array}\label{possibleErrors}\end{equation}
It's straightforward to check that these choices result in $a^g, b^g\in[0, 1]$ by letting: $$\Delta=\min\limits_{G}\left(\dfrac{\ell^g-m^g}{2\cdot \max\{\ell^g, 1-\ell^g\}}\right).$$ Then we have: $$0<\Delta\,\,\,\text{ and }\,\,\,\Delta\leq \left(\dfrac{\ell^g-m^g}{2\cdot \ell^g}\right)\,\,\,\text{ and }\,\,\,\Delta\leq \left(\dfrac{\ell^g-m^g}{2\cdot (1-\ell^g)}\right).$$ Plugging into $a^g$ and $b^g$ yields:
  $$\begin{array}{rcl}
a^g&=&\dfrac{1}{2}-\Delta \dfrac{1-\ell^g}{\ell^g-m^g},\\
&&\\
b^g&=&\dfrac{1}{2}+\Delta \dfrac{\ell^g}{\ell^g-m^g},
\end{array}$$ which satisfy $a^g, b^g\in[0, 1]$ by inspection. We can conclude that when $\delta^g, \mathcal{T}, \mathcal{F}, a^g$, and $b^g$ are defined as in Equations \ref{possibleClass}, \ref{possibleParameters} and \ref{possibleErrors}, both groups obtain classification outcomes conditional on behavior that equalize true positive and false positive error rates, so that error-rate balance is satisfied.

\ \\ \textbf{Step 2.} In Step 1 we applied a group-specific randomization of the raw signals to equalize  error rates across groups. In order to achieve predictive parity, we must now equalize rates of compliance across groups. By Equations \ref{eqThresh} and \ref{eqComply} we know that equilibrium compliance as a function of the classifier, or $\pi^g(\delta^g)$, depends on the true and false positive rates induced by $\delta^g$ ($\mathcal{T}$ and  $\mathcal{F})$, group-specific costs of compliance ($H^g$), and the net reward to a positive classification ($r^g$), with equilibrium compliance being:
$$\pi^g(\delta^g)=H^g\left(r^g\cdot(\mathcal{T}-\mathcal{F})\right).$$ Our construction of $\mathcal{T}$ and $\mathcal{F}$ in Step 1 ensured that $\mathcal{T}>\mathcal{F}$. Without loss of generality, let group $X$ have higher baseline prevalence, so that $H^X(0)\geq H^Y(0)$. We set $r^X=0$. For group $Y$ set: $$r^Y=\frac{\left(H^Y\right)^{-1}\left(H^X(0)\right)}{\mathcal{T}-\mathcal{F}}, $$ with $\left(H^Y\right)^{-1}$ the generalized inverse of $H^Y$. Because  $H^Y$ is continuous, such an inverse always exists (it need not be unique if $H^Y$ has ``flat spots''; in this case any selection from $(H^Y)^{-1}$ works for the construction). It follows that at classifier $\delta^g$ defined in Step 1, and setting $r^X=0$ and $r^Y$ as defined above, we have equalized prevalence across groups to $\pi$, with:
$$\pi= H^Y\left(r^Y\cdot(\mathcal{T}-\mathcal{F})\right)=H^X(0).$$ 

To summarize,  by Step 1 we designed $\delta^g$  to equalize error rates across groups. By Step 2 we equalized induced prevalence rates via differential rewards. Predictive parity requires equal positive predictive value across groups, or: $$\mathbb{E}\left[\beta|\delta^X\right]=\mathbb{E}\left[\beta|\delta^Y\right].$$ We have chosen $\delta^g$ and $r^g$ to ensure:

$$\mathbb{E}\left[\beta|\delta^g\right]=\frac{\pi\cdot\mathcal{T}}{\pi\cdot\mathcal{T}+(1-\pi)\cdot\mathcal{F}},$$ and because $\pi$, $\mathcal{T}$, and $\mathcal{F}$ are group-invariant,  this expression is identical across groups. Consequently, predictive parity is satisfied.    {Finally, we note that our construction sets} $r^X=0$    {for simplicity. Setting} $r^X$    {to any positive value and adjusting} $r^Y$    {accordingly yields an equivalent result, as prevalence equalization depends only on the relationship between} $r^X$    {and} $r^Y$    {rather than their absolute levels.}
\end{proof}

\ \\ \textbf{Corollary \ref{corollary}} \textit{Fix any informative classification rule that induces error-rate balance across groups. If the distribution of compliance costs for group $Y$ stochastically dominates that of group $X$ then predictive parity can be attained only by allowing the stakes of classification to differ across groups. Otherwise, there exists a classification rule and a system of stakes satisfying equal stakes, error-rate balance, and predictive parity, though potentially at the cost of violating aligned incentives.}

\ \\ \begin{proof}
Under error-rate balance, both groups share the same true- and false-positive rates $\mathcal{T}$ and $\mathcal{F}$, with $\mathcal{T} >\mathcal{F}$ (as our classifier is informative). For group $g$, we define the positive predictive value of a classifier as:
\[
\mathrm{PPV}^g
= \text{Pr}\left(\beta_i=1|d_i=1, g\right)
= \frac{\pi^g(\delta^g) \mathcal T}{\pi^g(\delta^g) \mathcal T + (1-\pi^g(\delta^g))\mathcal F},
\]
which is strictly increasing in equilibrium prevalence $\pi^g(\delta^g)$ whenever
$\mathcal T>\mathcal F$. Hence predictive parity holds if and only if $\pi^X(\delta^X)=\pi^Y(\delta^Y)$.

In equilibrium, prevalence is given by:
\[
\pi^g = H^g\big(r^g(\mathcal T-\mathcal F)\big).
\]
If equal stakes is satisfied, so that $r^X=r^Y=r$, then predictive parity requires:
\begin{equation}\label{PPV}
H^X\big(r\cdot(\mathcal T-\mathcal F)\big)
= H^Y\big(r\cdot(\mathcal T-\mathcal F)\big).
\end{equation}
If $H^Y$ strictly stochastically dominates $H^X$ then $H^Y(c)<H^X(c)$ for all costs $c$. Consequently Equation \ref{PPV} cannot be satisfied, and attaining predictive parity will require $r^Y>r^X$. 

Conversely, if  group cost distributions cannot be ordered via stochastic dominance then there exists at least one $\overline{c}$ with $H^X(\overline{c})=H^Y(\overline{c})$ and predictive parity can be achieved via equal stakes, assigning: $$r=\frac{\overline{c}}{\mathcal{T}-\mathcal{F}}.$$ However, if $\overline{c}<0$ this necessitates a violation of aligned incentives.
\end{proof}

\end{document}